\documentclass{aamas2014}
\usepackage{algorithm}
\usepackage{url}
\usepackage{amsmath}
\allowdisplaybreaks[4]

\newdef{definition}{Definition}
\newtheorem{theorem}{Theorem}
\newtheorem{lemma}{Lemma}
\newtheorem{corollary}{Corollary}
\floatname{algorithm}{Mechanism}

\begin{document}


\title{Mechanism design for resource allocation --\\
 with applications to centralized multi-commodity routing}



%
%
%
%

%

\numberofauthors{3}

\author{
%
\alignauthor
Qipeng Liu\\
       \affaddr{Institute of interdisciplinary}\\
       \affaddr{information sciences}\\
       \affaddr{Tsinghua University, China}\\
       \email{lqp1831@gmail.com}
\alignauthor
Yicheng Liu\\
       \affaddr{Institute of interdisciplinary}\\
       \affaddr{information sciences}\\
       \affaddr{Tsinghua University, China}\\
       \email{61c@live.cn}
\alignauthor Pingzhong Tang\\
       \affaddr{Institute of interdisciplinary}\\
       \affaddr{information sciences}\\
       \affaddr{Tsinghua University, China}\\
       \email{kenshinping@gmail.com}
}

\maketitle

\begin{abstract}
We formulate and study the {\em algorithmic mechanism design} problem for a general class of resource allocation settings, where the center redistributes the private resources brought by individuals. Money transfer is forbidden. Distinct from the standard literature, which assumes the amount of resources brought by an individual to be public information, we consider this amount as an agent's private, possibly multi-dimensional type. Our goal is to design truthful mechanisms that achieve two objectives: maxmin and Pareto efficiency.

For each objective, we provide a reduction that converts {\em any optimal algorithm} into a {\em strategy-proof mechanism} that achieves the {\em same objective}. Our reductions do not inspect the input algorithms but only query these algorithms as oracles. 

Applying the reductions, we produce strategy-proof mechanisms in a non-trivial application: network route allocation. Our models and result in the application are valuable on their own rights.
\end{abstract}








\section{Introduction}
One of the most important problems at the intersection of economics and computation is {\em algorithmic mechanism design},
which dates back to the seminal work of Nisan and Ronen
~\cite{Nisan99}.
The basic problem asks:

{\em  Given an algorithmic optimization problem, is it possible to efficiently produce a truthful mechanism that (approximately) achieves the optimal value of the original problem?}

Over the past decade, there has been a number of breakthroughs regarding this problem in settings where money transfer is allowed, including the rich literature on truthful welfare-maximizing mechanism design~\cite{Papadimitriou2008,Dughmi2010,Dobzinski2012}, and Bayesian incentive compatible (BIC) mechanism design~\cite{HartlineB2010}. Recently, the same problem has been investigated under the context of revenue optimal mechanism design~\cite{Cai2012,Cai2013}.

Distinct from the above literature, we study algorithmic mechanism design {\em without money}~\cite[Chapter 10]{nisan07a}. In particular, we study a general class of resource allocation problems, where each agent brings a certain amount of resources and the mechanism distributes these resources to achieve certain objectives. It is important to note that our setting differs from the standard resource allocation literature (such as one-sided matching, hedonic games, etc) in that each agent's type is the amount of resources she brings, rather than her preference over allocations. Our goal is to design strategy-proof mechanisms that achieve two objectives: maxmin and Pareto efficiency.

Our framework is rich enough to encompass, or at least heavily intersect with, a variety of applications, such as cloud resource allocation~\cite{Ghodsi2011,Parkes2012}, facility location~\cite{Procaccia2009}, fair division~\cite{Procaccia2013} and network route allocation [this paper].

\subsection*{Main results}

We make the following contributions.

\begin{enumerate}
\item {\bf Black-box reductions}

For any resource allocation problem, we provide two reductions, one for maxmin and the other for Pareto efficiency, that automatically convert any optimal algorithm into a strategy-proof mechanism that achieves the same objective. Our reductions do not make use of any algorithm but only assume black-box access to the algorithm.

\begin{itemize}
\item For the maxmin objective, we construct a polynomial time algorithm that, for any input, generates the optimal group-strategy-proof mechanism by calling the optimal algorithm only once.
\item For Pareto efficiency, we show that, if there is an algorithm that {\em serially optimizes} the utility profile, the algorithm {\em per se} is strategy-proof.
\end{itemize}

\item {\bf Application}

Simple reductions as they may seem, their generalities are demonstrated by a complex, practical application: network route allocation.

\noindent{\bf Network route allocation}. In this application, each node downloads a file from a server node located somewhere in the network. It can download via a direct route to the sever (it can do so without joining the mechanism, imposing an IR constraints to the design problem) or, by joining the mechanism, download via certain indirect routes that pass through other nodes. Each node has a certain private capacity (resource) that specifies the maximum limit of flow that can pass through it. Upon receiving the reports of private capacities, the mechanism returns a multi-commodity flow on the network. Our model encompasses both client-server networks and peer-to-peer networks~\cite[Chapter 4]{Parkes13}. Applying our reductions, one can obtain IR and strategy-proof mechanisms that achieve either maxmin or efficiency.

\end{enumerate}

\subsection*{Related work}

In the literature of mechanism design without money, a popular line of work concerns how to locate a facility~\cite{Miyagawa2001,Procaccia2009}. They study a mechanism design problem where agents are located on the real line and the mechanism select the location of a facility, where each agent's cost is its distance to the facility. The objectives are {\em maxmin} as well as {\em efficiency}. They give tight bounds on ratio of the optimal strategy-proof mechanism over the optimal algorithm without incentive constraints. This line of work has been extended to a number of variations, such as there are two facilities~\cite{Lu2010}, or the utility function is not linear~\cite{Fotakis13}, or the objective is the least square~\cite{Feldman2013}.

Resource allocation is also of central importance to the multiagent system community. It has been widely discussed in the application of smart grid~\cite{Rogers2012,Ramchurn2012} and has been of the main application scenario of security games (see~\cite{Yang2014}).

Our resource allocation notion is related to the classic bilateral trade setting~\cite{Myerson1983}, resource exchange setting~\cite{Barbera1995}, as well as the recently coined reallocation settings~\cite{Blumrosen2014}. A key distinction is, as mentioned, our setting treats the amount of contributed resources as private type while this amount is public information in their settings and they all treat valuation functions as types. Furthermore, their objectives are focused on characterization of efficient truthful mechanisms while our objective is to reduce mechanism design to algorithm design. 

The second blackbox reduction in this paper, i.e., the one for the Pareto efficiency objective, makes use of algorithms that serially maximizes the utility profiles. This idea naturally relates our reduction to the (randomly) serial dictatorship literature~\cite{svensson1999strategy,bogomolnaia2004random,manea2007serial,aziz2013pareto}, which are known to be strategyproof and Pareto efficient. Note again that their strategyproofness is defined with respect to truthful announcement of preferences, rather than contributed resources. Another distinction is that the serial dictatorship literature is concerned with allocation of indivisible items while our focus is on divisible resources. This distinction is not essential though, if randomization are permitted.

\section{The resource allocation problem}
\label{sec:setting}
We now formulate the resource allocation problem. 

An environment specifies the parameters for the mechanism designer to operate.

\begin{definition}
An environment is a tuple $\{\mathcal{N}, \mathcal{S}, \mathcal{P}, \mathcal{O}, u\}$, where
\begin{itemize}
\item $\mathcal{N}$ denotes the set of $n$ agents,
\item $\mathcal{S}=\mathcal{S}_1 \times \mathcal{S}_2 \times \cdots \times \mathcal{S}_n$, where each $\mathcal{S}_i$ is the private type set of agent $i$,
\item $\mathcal{P}$ is a set of public information and resources shared by all agents,
\item $\mathcal{O}$ is the set of outcomes.
\end{itemize}
\end{definition}


By revelation principle~\cite{Myerson81}, one can without loss restrict attentions to the set of direct revelation mechanisms, which can be regarded as functions that maps agents' reported type profile and public info. into an outcome.


\begin{definition}
Given an environment,  a deterministic mechanism is a function $\mathcal{M} : \mathcal{S} \times \mathcal{P} \to \mathcal{O}$.
\end{definition}

Each agent $i$ in the environment comes with a private amount of resources $s_i\in \mathcal{S}_i$ and is asked by the mechanism to report this quantity. Upon receiving all inputs, the mechanism returns an outcome, i.e., an allocation of resources.


For example, consider a variant of the {\em dominant resource problem} defined in~\cite{Ghodsi2011}, $\mathcal{S}_i = {\mathbb{R}_{+}}^m$ where the $j$-th element in the vector is the amount of $j$-th resource that agent $i$ owns. Agent $i$ needs $c_{i, 1}$ units of the first type of resource and $c_{i, 2}$ units of
the second and $\cdots$ $c_{i, m}$ units of the $m$-th in order to conduct $1$ unit of job task. Thus, $\mathcal{P}$ is the set of all possible such $\{(c_{i, 1}, c_{i, 2}, \cdots, c_{i, m}) \}_{j=1}^{n}$ and $\mathcal{O}$ is the set of all possible distributions of resources.

A resource allocation environment imposes certain feasibility constraints on any mechanism defined on it. For example, no mechanism shall allocate $3$ CPU units to an agent if there are only $2$ units of CPU within the society.

\begin{definition}
For any mechanism input $s_1 \in \mathcal{S}_1, s_2 \in \mathcal{S}_2, \cdots, s_n \in \mathcal{S}_n, p \in \mathcal{P}$,
define $FEA(s_1, s_2, \cdots, s_n, p)\subseteq \mathcal{O}$ as a set of {\em feasible outcomes} under $(s_1, s_2, \cdots, s_n, p)$.

\end{definition}

Consider again the dominant resource problem, $o \in FEA($ $s_1, \cdots, s_n, p)$ if and only if the resource allocation prescribed by $o$ is feasible\footnote{Feasibility means does not over-allocate any type of resource.} under the input.

We consider environments and feasibility constraints that satisfy the following resource monotone property.


\begin{definition}
\label{resource-pro}
For an environment $\{\mathcal{N}, \mathcal{S}, \mathcal{P}, \mathcal{O}\}$, a feasibility constraint fuction $FEA$ is {\em resource monotone} if $\forall i$, there is a partial order $\leq_i$ on $\mathcal{S}_i$, such
    that given any input profile $s_1, s_2, \cdots, s_n$ and $p\in \mathcal{P}$, for any $s'_i \leq s_i$, we have
    $$FEA(s_1, \cdots, s'_i, \cdots, s_n, p) \subseteq FEA(s_1, \cdots, s_i, \cdots, s_n, p).$$
\end{definition}


Intuitively, resource monotonicity states that the larger amount of resource one contributes, the larger is the set of feasible allocations. For a mechanism defined on an environment satisfying the resource monotonicity, we call it a {\em resource allocation mechanism}.


%
%

Every player has a utility function $u_i : \mathcal{O} \times \mathcal{P} \to \mathbb{R} $.
 ~In this paper, we consider a type of $u_i$ that must only be related to public information and outcome. This means that the agent's utility does not depend on her private resources. In addition, we assume that there exists an {\em empty outcome} $o^*$ that
for any other outcome $o$, $u_i(o^*, p) \leq u_i(o, p)$.

The above definition of utility is not uncommon in the resource allocation domain. Consider again the dominant resource problem, the number of total units of computation player
$i$ can conduct only depends on the allocation outcome $o$ and her $c$-vector $(c_{i, 1}, c_{i,2 }, \cdots, c_{i, m})$
and the empty outcome is the allocation where no one gets any resource.

The goal of resource allocation mechanism design is to optimize a certain real-valued function $w : \mathcal{O} \to \mathbb{R}$ for equilibrium outcomes.  Here $w$ can be different functions according to
different application scenarios. For example, in some cases $w$ denotes social welfare, i.e., $w(o) = \sum_{i=1}^n u_i(o, p)$; while in some other cases, $w$ denotes the minimum utility among all agents', i.e.,
$w(o) = \min_i u_i(o, p)$.

The solution concept in this paper is the dominant strategy equilibrium.
\begin{definition}
  A mechanism is {\em dominant-strategy truthful, (aka. strategy-proof)} if for any $s_1, \cdots, s_n, p$ and any $s'_i \leq_i s_i$\footnote{We assume throughout that  agents never overreport. Overreport can lead to infeasible allocations, which can be easily detected and punished.}, then
\begin{equation*}
    u_i\left(\mathcal{M}(s_1 \cdots s_i, \cdots s_n, p), p\right)
    \geq u_i\left(\mathcal{M}(s_1 \cdots s'_i \cdots s_n, p), p \right)
\end{equation*}
\end{definition}

\begin{definition}
A mechanism is {\em group strategy-proof} if for any $s_1, \cdots, s_n, p$ and any $s'_1, \cdots, s'_n$ where each $s'_i \leq_i s_i$,  then
\begin{equation*}
   u_i\left(\mathcal{M}(s_1, \cdots, s_n, p),p \right) \geq u_i\left( \mathcal{M}(s'_1, \cdots, s'_n, p), p \right)
\end{equation*}
\end{definition}

It is easy to see that the definition above subsumes (is stronger than) another definition of group strategyproofness that no subset of agents could jointly deviate so that the resulting outcome is better off for anyone in this subset.


Our goal in this paper is to provide a tool to design strategy-proof mechanisms that optimize certain objectives. In particular, given an algorithm that optimizes a certain objective, we use this algorithm as a black-box and return a strategy-proof mechanism that optimizes the same objective.

\section{Max-Min}

In this section, we consider the problem of designing a strategy-proof mechanism that maximizes the minimal utility among all agents, i.e., $w(o) = \min_{i\in \mathcal{N}} u_i(o, p)$. We show that, given an algorithm that computes the $\arg \max_{o} w(o)$, we can
construct a strategy-proof mechanism with output $o'$ such that $w(o')  = w(o)$ for each input.

\begin{definition}
   An environment is continuous if for any $(s_1, $  $\cdots, s_n, p)$,
    let $o_w = \arg\max_o w(o)$ and for any $u'_i \leq u_i(o_w,p)$, there must exist some outcome
    $o' \in FEA(s_1, s_2, \cdots, s_i, \cdots, $ $s_n, p)$ such that $u_i(o', p) = u'_i$ and
    for any $j\ne i$, $u_j(o', p) = u_j(o_w,p)$.
\end{definition}

\begin{theorem}
\label{thm:1}
    If there exists an algorithm $\mathcal{A}$ that optimizes $w(o) = \min_i u_i(o, p)$ for some continuous, resource monotone environment,
    one can efficiently construct a strategy-proof mechanism $\mathcal{M}$ that optimizes $w(o)$.
\end{theorem}

This theorem states that for the resource allocation problems under consideration, designing strategy-proof mechanism with max-min objective is no harder than
the corresponding algorithm design problem.

We prove the following stronger theorem instead.
\begin{theorem}
\label{thm:gsp}
    If there exists an algorithm $\mathcal{A}$ that optimizes $w(o) = \min_i u_i(o, p)$ for some continuous , resource monotone environment,
    one can efficiently construct a group-strategy-proof mechanism $\mathcal{M}$ that optimizes $w(o)$.
\end{theorem}

Given an algorithm $\mathcal{A}$, we can construct $\mathcal{M}$ as follows:
\begin{algorithm}
\begin{enumerate}
\renewcommand{\labelitemi}{}
    \item on inputs $s_1, s_2,\cdots, s_n, p$, run $\mathcal{A}(s_1, \cdots, s_n, p)$ to get the outcome $o_0$;
    \item let $u^* \gets \min_i u_i(o_0, p)$;
    \item for $i = 1, 2, \cdots, n$:
        \begin{itemize}
            \item find $o_i$ such that for all $j \ne i$, $u_j(o_i, p) = u_j(o_{i-1}, p)$ and
            $u_i(o_i, p) = u^* \leq u_i(o_{i-1}, p)$;

        \end{itemize}
    \item outputs $o_n$;
\end{enumerate}
\caption{A group strategy-proof mechanism via $\mathcal{A}$} \label{alg:1}
\end{algorithm}

In other words, we find an outcome $o_n$ that brings down all agents' utilities to $u^*$, the maxmin value computed by  $\mathcal{A}$. 
\begin{proof}
First, let us analyze the time complexity of the mechanism. Denote $TIME(\mathcal{A})$ as the time complexity of algorithm
$\mathcal{A}$ and $TIME(F )$ as an upper bound of the time complexity to find such new outcome. So the time complexity
is $O\left( TIME(\mathcal{A}) + n TIME(F) \right)$.
The time of computing $o_i$ is often small enough, for example in our application, so that $\mathcal{M}$ has time complexity $O(TIME(\mathcal{A}))$.

Now we verify that $\mathcal{M}$ really outputs a feasible solution which optimizes $w(\cdot)$. The outcome $o_n$ is feasible
because in each step $o_i$ is feasible by the continuity of the environment.  Also, it is straightforward that $w(o_n) = w(o_0)$.

Finally, we prove that $\mathcal{M}$ is group-strategy-proof. Consider any input $s_1, s_2, \cdots, s_n, p$  and
everyone reports $s'_i$ where $s'_i \leq_i s_i$:

\begin{eqnarray*}
    &&u_i\left( \mathcal{M}(s_1, s_2, \cdots, s_i, \cdots, s_n, p), p \right) \\
    &=&\max_{o \in FEA(s_1, \cdots, s_i, \cdots, s_n, p)} w(o) \\
        &\geq& \max_{o \in FEA(s'_1, s_2, \cdots, s_i, \cdots, s_n, p)} w(o) \\
        &\,&  \vdots\\
        &\geq& \max_{o \in FEA(s'_1, s'_2,  \cdots, s'_i, \cdots, s'_n, p)} w(o) \\
        &=& u_i\left( \mathcal{M}(s'_1, s'_2, \cdots, s'_i, \cdots, s'_n, p) , p\right)
\end{eqnarray*}

So $\mathcal{M}$ is a group-strategy-proof mechanism that optimizes $w(o)$.
\end{proof}

Applying the same technique, one can produce a strategy-proof mechanism to optimize $w(o) = \min_i f_i(u_i(o, p))$ with
$n$ strictly increasing functions $f_1, f_2, \cdots, f_n$.  With this extension, we can optimize some interesting objectives subject to the {\em individual rationality} constraints. We now formally
define individual rationality.

\section{Max-min subject to individual rationality}


In this section, we consider the setting where if a player does not join the mechanism,
he will get a utility $r_i(p)$ that only depends on the public information $p$. In the previous setting, every
player's $r_i(p) = -\infty$; in other words, she always wants to join the mechanism in the previous setting.

\begin{definition}
    A mechanism is {\em individual rational} if on any inputs $s_1, s_2, \cdots, s_n, p$, the mechanism outputs $o$ such
    that for any $i \in \mathcal{N}$, $u_i(o, p) \geq r_i(p)$.
\end{definition}

\subsection{Replacing $\mathbf{u_i}$ by $\mathbf{f_i(u_i)}$  }

So the problem now is to design a strategy-proof and individual rational mechanism that achieves certain objective.
This can be achieved by the same method as Theorem~\ref{thm:1}.
\begin{theorem}
\label{thm:3}
    If there exists an algorithm $\mathcal{A}$ that computes $w(o) = \min_i f_i(u_i(o, p))$ for a continuous environment with
    strictly increasing functions $f_1, f_2, \cdots, f_n$, one can construct
     a strategy-proof mechanism $\mathcal{M}$ that computes $w(o)$.
\end{theorem}

The reduction is identical to that in theorem \ref{thm:1} thus the proof is omitted. 
%

By this theorem, we can design individual rational and strategy-proof mechanism that optimizes the following:
\begin{corollary}
    Let $f_i(x) = x - r_i(p)$, one obtains an individual rational and strategy-proof mechanism that optimizes the minimal utility gain $f_i$.
\end{corollary}
\begin{corollary}
    If $r_i(p) > 0$ for any $i \in \mathcal{N}$,
    let $f_i(x) = x / r_i(p)$, one obtains an individually rational and strategy-proof mechanism that optimizes minimal
    increasing rate.
\end{corollary}

\subsection{Optimizing $\min_i u_i(o, p)$ subject to IR}
A further question is, can we optimize $\min_i u_i(o, p)$ instead of $w(o) = \min_i f_i(u_i(o, p))$, subject to IR? We answer this affirmatively by modifying the previous method to produce an IR and SP mechanism that optimizes  $\min_i u_i(o, p)$.

\begin{theorem}
    \label{ir+sp}
    If there exists an algorithm $\mathcal{A}$ that computes $w(o) = \min_i u_i(o, p)$
    for any continuous, resource monotone environment,
    such that $u_i(o, p) \geq r_i(p)$,
    then we can construct a strategy-proof and individual rational mechanism $\mathcal{M}$ that computes $w(o)$.
\end{theorem}

\section{Pareto efficiency}

In this section, we discuss the objective of implementing a {\em serially optimal} outcome.
We still use the same setting except that we now restrict agent's $i$ type space $\mathcal{S}_i$ to $ \mathbb{R}$.

\begin{definition}
    An outcome $o$ is called serially optimal under $s_1, \cdots, s_n, p$ if the utility profile $(u_1(o, p), \cdots, u_n(o, p))$ is the
    lexicographically largest one among all possible utility profiles; formally, for any $o' \in FEA(s_1, \cdots, s_n, p)$,
    either the utility profile of $o$ equals to that of $o'$ or there exists some $j \in \mathcal{N}$ such that
    for all $k < j$, $u_k(o, p) = u_k(o', p)$ and $u_j(o, p) > u_j(o', p)$.
\end{definition}

Clearly, serial optimality implies Pareto efficiency, so we focus on the former.
Our problem now becomes, given any list of input, compute a strategy-proof and serially optimal outcome.  Before we start, we need the following property.

\begin{definition}
    The utility functions are {\em monotone} if 
    for any $i \in \mathcal{N}$, any inputs $s_1, \cdots, s_n, p$
    and $0 < \delta \leq u_i(o,p) - r_i(p)$,
    there exists an $\varepsilon > 0$ such that for any $0 < \varepsilon' < \varepsilon$, there exists an outcome $o' \in FEA(s_1, \cdots, s_i - \varepsilon', \cdots, s_n, p)$ such that for any $j\ne i$, $u_j(o', p) = u_j(o, p)$ and $u_i(o', p) \geq u_i(o, p) - \delta$.
\end{definition}

\begin{theorem}
\label{pe+sp}
If there exists an algorithm $\mathcal{A}$ that serially optimizes $w(o) = (u_1(o, p), u_2(o, p), \cdots, u_n(o, p))$ for monotone
utility functions, resource monotone environment and $u_i(o, p) \geq r_i(p) $, 
and for each player $i$, her utility, which is a function of $s_i$, $u_i\left( \mathcal{A}(s_1 \dots s_i \dots s_n, p) , p\right)$ is
continuous, 
the algorithm itself is an IR and SP mechanism that computes a serially optimal outcome.
\end{theorem}

\begin{proof}

\begin{figure}[htbp]
\centerline{\includegraphics[scale=0.8]{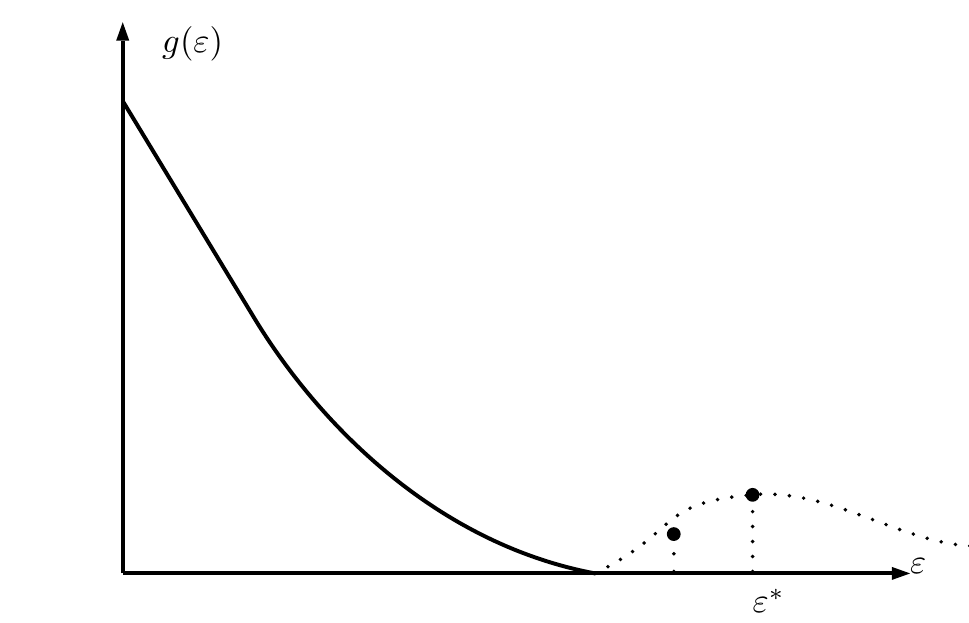}}
\caption{A figure for our proof}
\label{fig:one}
\end{figure}

    We denote $x_i(\varepsilon) = u_i\left(\mathcal{A}(s_1, \cdots, s_i - \varepsilon, \cdots, s_n, p), p\right)$ for $0 \leq \varepsilon \leq s_i$,
    for fixed inputs $s_1, \cdots, s_n, p$ (we call it (1)). $x_i(\varepsilon)$ is a continuous function of $s_i$. 
    First we prove that there exists an $\varepsilon > 0$,
    $\mathcal{A}$ outputs an outcome $o'$ such that $u_i(o', p) \leq u_i(o, p)$ under
    $s_1, \cdots, s_i - \varepsilon', \cdots, s_n, p$ (we call it (2)) for  any $0 < \varepsilon' \leq \varepsilon$.

    We prove it by contradiction. An important observation is that any feasible solution under (2) is also feasible under (1) by the resource monotone environment condition,
    in other words, $$FEA(s_1, \cdots, s_i - \varepsilon', \cdots, s_n, p) \subseteq FEA(s_1, \cdots, s_i, \cdots, s_n, p).$$

    If $\left(u_j(o', p)\right)_{j=1}^n$ has the property that $u_i(o', p) > u_i(o, p)$,
    then there must exist $k <i$ such that $u_j(o', p) = u_j(o, p)$ for all $j < k$ and $u_k(o', p) < u_k(o, p)$
    (because a solution in (2) is a solution in (1) and $\left(u_j(o, p)\right)_{j=1}^n$ is the serially optimal). However we can construct a larger solution
    for (2): by the monotone utility function condition, we know that there exists an outcome $o''$ such that $u_j(o'', p)= u_j(o, p)$ for all $j\ne i$ and
    $u_j(o'', p) > u_i(o, p) - \delta$ for some $\delta$ and $\varepsilon$ which is a contradiction that the algorithm find a serially optimal solution.

    Next we prove that $u_i(\varepsilon)$ is a decreasing function. Let $g_i(\varepsilon) = x_i(\varepsilon) - r_i(p)$. Because the algorithm is IR,
    $g_i(\varepsilon) \geq 0$ for all feasible $\varepsilon$.   For any $\varepsilon$ such that $g_i(\varepsilon) > 0$, there exists an $\delta > 0$ such that
    $g_i(\varepsilon) \geq g_i(\varepsilon')$ for any $\varepsilon \leq \varepsilon' \leq \varepsilon + \delta$ (by the monotone utility function) so it is
    locally monotonically decreasing. When $g_i(\varepsilon) = 0$, for all feasible $\varepsilon' > \varepsilon$, $g_i(\varepsilon') = 0$. If $g_i(\varepsilon') > 0$, we can find a
    first extreme point $\varepsilon^* \in (\varepsilon, \varepsilon']$. Then for any small $\delta > 0$, $g(\varepsilon^* - \delta) < g_i(\varepsilon^*)$ which conflicts with
    the utility function monotone.

    So $g_i(\varepsilon)$ is a monotonically decreasing function when $g_i(\varepsilon) > 0$. And if $g_i(\varepsilon^*) = 0$, then $g_i(\epsilon')$ for $\epsilon' > \epsilon^*$ will always be 0 .
    We can conclude that $g_i(\cdot)$ is a monotonically decreasing function. Since $x_i(\cdot) = g_i(\cdot) + r_i(p)$, $x_i(\cdot)$ is also a monotonically decreasing
    function which means player $i$ will never get more profit by reporting a lower $s'_i \leq s_i$.
\end{proof}
The theorem states that for monotone utility functions and resource environment, mechanism design is as easy as algorithm design.

\subsection*{Serial optimization by contribution}

In order to do serial optimization, one needs to pre-specify a (static) ordering on $\mathcal{N}$, resulting in unfairness for agents who have low rank. In this section, we show that, without sacrificing strategyproofness, this ordering can be extended to depend on the input. In particular, this ordering can be consistent with the ranking of agents' contributions: the more one contributes, the higher priority she gets.

Formally, define ``order envy-freeness'' as follows:
\begin{definition}
    Consider the class of serially optimized mechanisms, a mechanism is order envy-free (OEF) with respect to
    $\{l_i\}_{i=1}^n$ ($l_i$ is a weight for agent $i$)
    if for any two agents $i, j \in \mathcal{N}$, such that $l_i s_i > l_j s_j$,  $i$ is ranked before $j$ in the serial ordering.
\end{definition}

Order envy-freeness states that an agent gets a higher priority for optimization than the ones with lower contributions. In particular, if set $l_i=1 ~\forall i$, the order of optimization is same as the ranking on reports.

\begin{theorem}
\label{thm:oef}
\label{oef}
    If there exists an algorithm $\mathcal{A}$ that serially optimizes $w(o, q) = (u_{q_1}(o, p), u_{q_2}(o, p), \cdots, u_{q_n}(o, p))$
     ( where $q$ is an index on agents such that $l_{q_i} s_{q_i} \geq l_{q_j} s_{q_j}$ for any $i<j$) for monotone utility functions, resource monotone environment and
    $u_i(o, p) \geq r_i(p)$, then the algorithm itself is an IR, SP and OEF (w.r.t $\{l_j\}_{j=1}^n$) mechanism that compute a serially
    optimal outcome.
\end{theorem}


To achieve order envy-freeness, we define a slightly different reduction as follows. Given reported $\{s_i\}$,
we generate an ordering $q_i$ such that for any $i$, $l_{q_i} s_{q_i} > l_{q_{i+1}} s_{q_{i+1}}$ and optimizes each $x_i$ according to
the order $\{q_i\}$.  In other words, we define the order according to their contributions: the more you share, the earlier you get served.
Technically speaking, this modified mechanism does not belong to the class of mechanisms we have proposed since the ordering of optimization now explicitly
depends on the report. However, it is easy to check that this modification does not affect IR, PE, SP and guarantee OEF as a plus. The mechanism is listed below. The proof follows a similar argument to the static case, which we do not repeat.

\begin{algorithm}
\begin{enumerate}
\renewcommand{\labelitemi}{}
    \item on inputs $s_1, s_2,\cdots, s_n, p$, and $l_1, \cdots, l_n$,
    \item compute the order $q_i$ by any sorting algorithm,
    \item run $\mathcal{A}(s_{q_1}, \cdots, s_{q_n}, p, q)$ to get the outcome $o$;
    \item output $o$
\end{enumerate}
\caption{An IR, SP, PE and OEF mechanism}
\end{algorithm}

\section{Application: Network route allocation}

Starting from this section, we show the generality of our methodology by applying it to two realistic scenarios. The two applications are highly nontrivial and valuable on their own rights.


We first consider an interesting route allocation problem in multiple-commodity networks. For the network structures under consideration, there are several vertices known as users. Each user demands a file of certain size stored on a server. The users form a directed graph. Each user has a capacity, which denotes the maximum (traffic) flow that can go through that vertex. We assume this vertex capacity is private information of the user. Between each user and each server, there is an arc constrained by certain capacity, denoting the maximal flow that can go through the arc. Given such a network, a user can download her target file via any route to the destination server\footnote{For each user, there exists a direct route via which the user can down its file without even join the mechanism. By joining the mechanism, however, the user can download the file via any indirect route that pass through other uses who also join the mechanism by sharing their bandwidths. This outside option imposes an IR constraint for the design problem.}. Given the reported vertex capacities, a route allocation mechanism allocates a route (or multiple routes, both of which we consider) and feasible flow within the route for each user.

It is not hard to see that our formulation encompasses route allocation problems in both client-server based and peer-to-peer based (simply treat a server in our model as a peer that never downloads) networks. In practice, the same route allocation problem has been witnessed by Xun-You Inc.\footnote{\url{http://www.xunyou.com/}}, an online gaming platform that aims to resolve the congestions on networks consisting of subnetworks by multiple Internet Service Providers. The same problems have also been witnessed by route optimization among multiple express companies, each of which specializes in some geographic region.

 A user's utility is the negation of its time delay, given by $$-\frac{\mbox{file size}}{\mbox{allocated flow}}.$$ This notion of utility is widely used in the evaluation of performance in a network system or an operating system. More importantly, it does not depend on the private information of users, i.e., the private capacity of each vertex, allowing our framework to be applicable.

We will use the above reduction theorems to design a strategy-proof mechanism that achieves the \emph{minmax} objective (see~\cite{Nisan99}), i.e., to minimize the maximum delay among all users. We also consider the objective of Pareto efficiency.

\subsection{Setting}

Formally, in our setting, there are $n$ users $ P = \{p_1, p_2, \cdots , $ $p_n\}$ and $m$ servers
$Q = \{q_1, q_2,\cdots , q_m\}$ ($p_i$ is called user $i$ and $q_j$ is called server $j$).
For a user $p_i$, her desired file is stored on server $d_i$, i.e, $q_{d_i}$.
The size of the file is $c_{i}$ units. If the amount of
her available bandwidth\footnote{We use ``bandwidth'' and ``flow'' interchangeably.} is $x$ ($x$ units of flow per second), she can finish her task in $c_i / x$.

The connections between users form a directed graph $G = (V , E)$,
where $V$ is the set of vertex in the graph, $E$ is the set of edges.
For a vertex in $V$, it denotes a user in $P$.
A pair of users can connect with each other.
For an edge $\langle u, v \rangle \in E$, we assume there is not limitation in the edge
$\langle u, v \rangle$, in other words, arbitrarily large flow can pass through this edge. This is with out loss of generality --- all results carry over to the setting where each edge has a capacity.

Also, each user $p_i$ can connect to a subset of servers. For a pair of user $p_i$
and server $q_j$, there is an edge between them with bandwidth $b_{i, j}$. That is, this edge can transmit at most $b_{i, j}$ units of flow per second. When
$b_{i, j} = 0$, there is no connection between $p_i$ and $q_j$.

Each user has a bandwidth limitation locally. The total
amount of flow she can download and share with others per second is at most $v_i$.
In reality, $v_i$ is the limitation that can depend on user $i$'s hardware or software.
A user $p_i$ can download her own file either by the edge directly between $p_i$ and $q_{d_i}$ or
indirectly from some $p_j$'s where $p_j$ downloads it from the edge $p_j$ and $q_{d_i}$.

\begin{figure}[!t]
\centering
\includegraphics[scale=0.8]{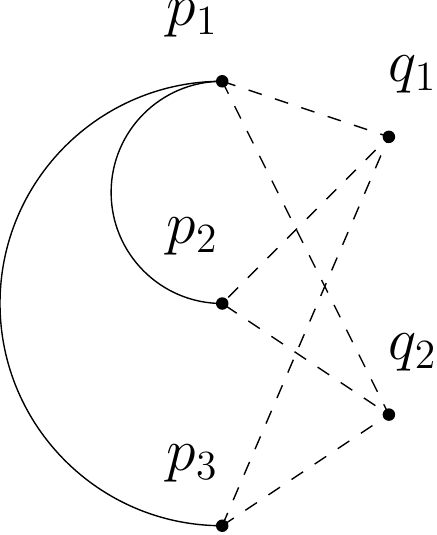}
\caption{A sample model, $d_1 = d_2 = 1$, $d_3 = 2$}\label{fig1}
\end{figure}

Consider the illustrating example in Figure \ref{fig1}.
The left nodes are users and the right nodes are servers. The solid lines denote connections between users and dashed lines denote connections between users and servers.
In this example, user $p_3$ could download its file through the edge $\langle p_3, q_2 \rangle$, through the path
$\langle p_3, p_1\rangle, \langle p_1, q_2 \rangle $ or the path $\langle p_3, p_1 \rangle ,
\langle p_1, p_2\rangle,$ $ \langle p_2, q_2\rangle$.

We assume that $G, \{b_{i, j}\}_{  \begin{subarray}{c} i=1\dots n \\ j=1 \cdots m\end{subarray}    }$ , $\{(d_i, c_i)\}_{i=1\cdots n}$ are public information that is shared by all users and the mechanism designer, while
$\{v_i\}_{i=1}^n$ is private information known only to each user $i$.
Given a reported profile of $(v_1, v_2, \cdots, v_n)$, a mechanism returns a flow assignment $x_i$ for each $i$ (a multiple-route assignment).
%
%
%

Thus, we can define the environment $\left\{\mathcal{N},\mathcal{S},\mathcal{P},\mathcal{O},u\right\}$ of this problem:
\begin{itemize}
    \item $\mathcal{N}$ denotes the set of $n$ agents,
    \item $\mathcal{S} = \mathcal{S}_1 \times \mathcal{S}_2 \times \cdots \mathcal{S}_n$ each $\mathcal{S}_i = \mathbb{R}$
    represents the local bandwidth limitation $v_i$,
    \item $\mathcal{P}$ represents servers $Q$, the whole graph structures $G = (V, E)$, $\{(d_i, c_i)\}_{i=1\cdots n}$  and all information about capacities on edges,
    \item $\mathcal{O}$ is the set of all the possible assignments,
    \item $u_i(o, p) = - \frac{c_i}{x_i}$ where $x_i$ is $i$'s total amount of flow in the assignment $o$.
\end{itemize}

Also we need to define individual rationality $r_i(\cdot)$:
for any agent $i \in \mathcal{N}$, the flow she receives by a mechanism is greater than or equal to
her direct route bandwidth $b_{i, d_i}$ if she does not join the mechanism. Formally, $r_i(p) =  b_{i, d_i}$.

\subsection{Max-Min objective subject to individual rationality}

\begin{lemma}
    The above environment is resource monotone and continuous.
\end{lemma}
\begin{proof}
For any feasible assignment in
$FEA(s_1, s_2, \cdots, s_n, p)$, it doesn't violate any constraint in $FEA(s_1, \cdots, s_i + \varepsilon, \cdots, s_n, p)$,
implying monotonicity.

For any $u'_i \leq u_i$, we can simply decrease $x_i$ to $x'_i = -c_i/u'_i \leq -c_i/u_i = x_i$
by cancelling $x_i - x'_i$ units of flow in the assignment.  So it is continuous.
\end{proof}

By theorem \ref{ir+sp}, we only need to give an algorithm to optimize $\min_i u_i(o)$. Because the multiple-commodity flow problem can be written
in linear program, so we can get the algorithm by binary search and linear programming:

\begin{algorithm}
\label{findD}
\begin{itemize}
\renewcommand{\labelitemi}{}
    \item binary search $D \in \left[0, \max_{i \in \mathcal{N}} \frac{c_i}{b_{i, d_i}} \right]$, if the following LP is feasible under
certain $D$, decrease $D$; otherwise, increase $D$ :
\item
\begin{itemize}
\renewcommand{\labelitemi}{}
\renewcommand{\labelitemii}{}
    \item  maximize : 0 (No objective here, i.e., a linear-feasibility program.)
    \item subject to
            \begin{enumerate}
                \item $x_i = x_{i, 1} + x_{i, 2} + \cdots + x_{i, n}$ for all $i = 1, \cdots, n$;
                \item $x_i = x_{i, i} + \sum_{e'} f_{i, e'}$ ($e'$ going into $i$) for all $i = 1, \cdots, n$;
                \item $\sum_{e''} f_{i, e''} = \sum_{e'} f_{i, e'} + x_{i, j}$ ($e'$ going into $j$ and $e''$ going out of $j$)
                for all $i\ne j$;
                \item $x_i \geq b_{i, d_i}$ for all $i$;
                \item $\mathbf{x_i\geq c_i/D}$ \textbf{for all} $\mathbf{i}$
                \item $\Sigma_jx_{j,i}+\sum_{e', j} f_{j, e'} \leq v_{i}$ ($e'$ going into of $i$) ;
                \item $\sum_{k} x_{k, i} \leq b_{i, j}$ (where $k$ have the following property: $d_k = j$), for all $i, j$;
                \item all variables are non-negative.
            \end{enumerate}
\end{itemize}
\end{itemize}
\caption{An algorithm to minimize the maximal cost}
\end{algorithm}

\begin{enumerate}
    \item In the first equations, $x_i$ denotes the final flow user $i$ can get,
and $x_{i, k}$ denotes the flow user $i$ gets from the edge $\langle p_k, q_{d_i} \rangle$. $x_{i, i}$ is the flow
on the direct edge from the user $i$ to the server.
    \item Let $f_{i, e}$ denote the final flow in edge $e$
    when transporting the file $c_i$ between $d_i$ and $p_{d_i}$. So, for the second equation,
$x_i$ is the sum of all $f_{i, e'}$ (ingoing flow) and $x_{i, i}$.
    \item The third one shows that in each vertex, the
total ingoing flow equals the total outgoing flow, which is known as the ``flow conservation'' constraints. Note that when $i = j$, the
equation does not make sense, since it is impossible for user $i$ to transport her target flow to others.
    \item The fourth constraints ensures individual rationality.

    \item The fifth means the upper bound of time each user needs is at most $D$.

    \item The sixth inequalities
$\sum_j x_{j, i} + \sum_{e', j} f_{j, e'} \leq v_{i}$ are the `capacity constraints' on vertices.

    \item The sixth inequalities $\sum_{k} x_{k, i} \leq b_{i, j}$ are the `capacity constraints' on edges.

    \item The final constraints ensures all variables being non-negative.
\end{enumerate}

With Algorithm~\ref{findD}, we obtain a SP, IR and Maxmin mechanism (Algorithm~\ref{maxmin_network}) by theorem \ref{ir+sp}:
\begin{algorithm}
\label{maxmin_network}
\caption{A strategy-proof minmax Mechanism}
\begin{enumerate}
    \item Binary search the minimized maximal downloading time $D^{\ast}$ among the solution of the LP (listed above)
    \item For all $p_i\in P$
    \begin{enumerate}
        \item if $\frac{c_i}{b_{id_i}}<D$, set $x_i'\gets b_{id_i}$
        \item else set $x_i' \gets \frac{c_i}{D}$
    \end{enumerate}
    \item For all $p_i\in P$, add constraint $x_i=x_i'$
    \item Solve the new LP to compute the assignment.
\end{enumerate}
\end{algorithm}

\subsection{Pareto efficiency subject to individual rationality}

By theorem \ref{pe+sp}, we know that if there exists an algorithm to find a serially optimal
outcome $o$ for monotone utility function, resource monotone environment and IR,
it is a strategy-proof mechanism.  Algorithm \ref{algo:1} gives the outcome we want.
Almost all constraints are the same as what we listed in the previous section except that
the seventh constraint ensures that the final solution $(u_1(o, p), \cdots , u_n(o, p))$ is the serially optimal one
among all the outcomes.

\begin{algorithm}
Given $G, \{b_{i, j}\}, \{(d_i, c_i)\}, \{v_i\}$, repeat the following procedure $n = |V|$ times:
\label{algo:1}
\begin{itemize}
\renewcommand{\labelitemi}{}
    \item  maximize : $x_i$ (the $i$-th time)
    \item subject to
            \begin{enumerate}
                \item $x_i = x_{i, 1} + x_{i, 2} + \cdots + x_{i, n}$ for all $i = 1, \cdots, n$;
                \item $x_i = x_{i, i} + \sum_{e'} f_{i, e'}$ ($e'$ going into $i$) for all $i = 1, \cdots, n$;
                \item $\sum_{e''} f_{i, e''} = \sum_{e'} f_{i, e'} + x_{i, j}$ ($e'$ going into $j$ and $e''$ going out of $j$)
                for all $i \ne j$;
                \item $x_i \geq b_{i, d_i}$ for all $i$;
                \item $\sum_j x_{j, i} + \sum_{e', j} f_{j, e'} \leq v_{i}$ ($e'$ going into  $i$) ;
                \item $\sum_{k} x_{k, i} \leq b_{i, j}$ (where $k$ have the following property: $d_k = j$), for all $i, j$;
                \item for all $j \leq i - 1$, set $x_i = x^*_i$ where $x^*_i$'s are the $x_i$'s optimized so far;
                \item all variables are non-negative;
            \end{enumerate}
\end{itemize}
\caption{An IR, PE and SP mechanism}
\end{algorithm}

We have already checked the resource monotone condition and
still need to check whether the utility function is monotone.

\begin{lemma}
\label{monotone-utility}
    $\{u_i\}_{i=1}^n$ are monotone utility functions under the environment.
\end{lemma}

To prove the lemma \ref{monotone-utility},  we need to prove the following lemma.
\begin{lemma} \label{alluse}
    For any possible IR solution $\{x_1, x_2, \cdots, x_n\},\, $ $ \{x_{i,j}\}, \{f_{i,e}\}$, there exists an IR solution
    $\{x_1, \cdots, x_n\}, \{x'_{i, j}\}, $ $ \{f'_{i,e}\}$ with $x'_{i, i} = b_{i, d_i}$.
    \end{lemma}

    The lemma states that, it is without loss of generality to restrict to allocations where each user exhausts all its bandwidth.

\begin{proof}
    First we prove that all bandwidth in $\langle i, d_i \rangle$ will always be
    used in some solution. Because of the inequalities $x_i \geq b_{i, d_i}$, the final bandwidth
    $x_i$ is at least $b_{i, d_i}$. If the bandwidth in the edge $\langle p_i, q_{d_i} \rangle$ is not fully exhausted
    (let  $b'$ be the remaining bandwidth in that edge), we can remove $b_{i, d_i} - b'$ bandwidth from other paths and add it into the
    edge $\langle i, d_i \rangle$. This operation does not violate any constraint.
    So we can repeat this operation and finally all $\langle i, d_i \rangle$ is fully exhausted.

    After having such a solution, we can modify some flow assignment so that
    all bandwidth in $\langle i, d_i \rangle$ will always be used by $i$ in some solutions.
    Now every edge $\langle p_i, q_{d_i} \rangle$ is fully used, we can use exchange techniques to let
    $\langle p_i, q_{d_i} \rangle$ is only used by $i$. If part of $\langle p_i, q_{d_i}\rangle$ is used by other user $j$,
    there must be a path from $p_i$ to $p_k$ and then from $p_k$ to $q_{d_i}$ with bandwidth $\varepsilon$.
    We can give this path to user $j$ (attached to the path from $p_j$ to $p_i$)
    and give $\varepsilon$ units bandwidth in $\langle p_i, q_{d_i}\rangle$ to $p_i$. Repeating this
    operations, we finally let all bandwidth in $\langle p_i, q_{d_i}\rangle$  be used only by user $i$.
\end{proof}
We now prove that the utility functions are monotone:
\begin{proof}
For any $i \in \mathcal{N}$ and any inputs $s_1, \cdots, s_n, p$, $i$ will get $u_i(o,p) = - \frac{c_i}{x_i}$ from
the best outcome $o$.
For any $0 < \delta \leq  - \frac{c_i}{x_i} + \frac{c_i}{b_{i, d_i}}  $,
we can construct an outcome $o'$ with utility profile $\{u_1(o, p), \cdots, u_i(o, p) - \delta, \cdots, u_n(o, p)\}$
from a solution for $o$.  A simple observation is that $x'_i = \frac{- c_i}{u_i(o, p) - \delta}$. So one straightforward
method is to decrease $x_i$ to $x'_i$.  Let $\varepsilon = x_i - x'_i$.
    Because of lemma \ref{alluse}, there is a solution $\{x_{j, k}\}, \{f_{j, e}\}$ to get the same utility profile of $o$
    such that $x_{j, j} = b_{j, d_j}$ and $\sum_{k=1}^n x_{j,k} = x_j$ for all $j$. Now we construct the
    following outcome $o'$:
    \begin{itemize}
        \item $x'_{i, i} = x_{i, i} - \varepsilon$, and for other $x'_{j, k} = x_{j, k}$;
        \item $x'_{i} = x_i - \varepsilon$, and for other $x'_j = x_j$;
        \item $f'_{j, e} = f_{j, e}$ for all possible $j, e$;
    \end{itemize}

    Now we check whether the solution violates any constraint. For the first two inequalities, they still
    hold.  For the third and sixth one, they aren't related to any new variables so it still holds.
    Because $  0 < \delta \leq -\frac{c_i}{x_i} + \frac{c_i}{b_{i, d_i}} $,
    we find that $x'_i = \frac{-c_i}{u_i(o, p) - \delta} \geq b_{i, d_i}$.

For the fifth one,
    \begin{eqnarray*}
        \sum_j x'_{i, j} + \sum_{e', j} f'_{j, e'} &=&
        x'_{i, i} + \sum_{j\ne i} x'_{i, j} + \sum_{e',j} f'_{j, e'}  \\
        &=& x_{i, i} - \varepsilon + \sum_{j\ne i} x_{i, j} + \sum_{e',j} f_{j, e'}  \\
        &\leq& v_i - \varepsilon
        = v'_i
    \end{eqnarray*}

    Then, we will show that $x_1(\cdot), \cdots,$ $ x_n(\cdot)$ is continuous functions of variable $\varepsilon$. We can use the
    following lemma~\cite{adler199tgeometric}: for the function $x_1(\varepsilon), $ there exists a closed interval $[\alpha, \beta]$,
    such that the LP is infeasible for all $\varepsilon \not\in [\alpha, \beta]$; and in the
    interval $[\alpha, \beta] $, $x_1(\varepsilon) \subseteq \mathbb{R}$ and
    it is a continuous convex piecewise linear function. So for $x_1(\cdot)$, it is a continuous function.

    Now using induction, if we have proven that $x_1(\cdot) , \cdots, x_k(\cdot)$ is continuous functions, then
    $x_{k+1}(\varepsilon) = x_{k+1}(\varepsilon, x_1(\varepsilon), \cdots, $ $x_k(\varepsilon))$.
    Using the theorem about multiparametric linear programming~\cite{gal1972multiparametric}, we know that $x_{k+1}$ is a
    continuous function over $\{\varepsilon, x_1 ,\cdots, x_k\}$. Also because $x_i(\cdot)$ is a continuous function
    of $\varepsilon$, it is easy to verify $x_{k+1}(\cdot)$ is a continuous function of $\varepsilon$ according to the definition.

    All the constraints hold. So these utility functions are monotone.
\end{proof}

Lemma \ref{monotone-utility} implies algorithm \ref{algo:1} is strategy-proof.


In addition,  we can obtain order envy-freeness by sorting agents according to weights $\{l_j\}$  and types $\{s_j\}$ and optimizing
them one by one. This is quite straightforward by theorem \ref{oef};

To conclude the section, we have the following theorem:
\begin{theorem}
    There exists an IR and SP mechanism which minimizes maximal cost/delay and
    an IR, OEF and SP mechanisms which gives a Pareto efficient solution for the route allocation problem.
\end{theorem}

\section{Conclusion and futher work}

In this paper, we studied truthful mechanism design for a general class of resource allocation settings, where the center redistributes the private resources brought by individuals. We designed truthful mechanisms that achieve two objectives: maxmin and Pareto efficiency. For each objective, we provided a reduction that converts {\em any optimal algorithm} into a {\em strategy-proof mechanism} that achieves the {\em same objective}. Applying the reductions, one can systematically produce strategy-proof mechanisms in a non-trivial application: network route allocation.

There are a number of exciting research directions to extend this work. First of all, the mechanism for the maxmin objective is wasteful in the sense that all the players are brought down to a utility level that equals the maxmin. Is there another strategy proof mechanism that achieves maxmin but yields better social welfare? Second, one can consider the same set of objectives under more general settings where randomization is allowed. Last but not least, we are also interested in finding more applications to our general algorithmic mechanism design framework.

\bibliographystyle{abbrv}
\bibliography{bandwidth}

\begin{thebibliography}{10}

\bibitem{adler199tgeometric}
I.~Adler and R.~D. Monteiro.
\newblock A geometric view of parametric linear programming.
\newblock {\em Algorithmica}, 8(1-6):161--176, 1992.

\bibitem{aziz2013pareto}
H.~Aziz, F.~Brandt, and P.~Harrenstein.
\newblock Pareto optimality in coalition formation.
\newblock {\em Games and Economic Behavior}, 82:562--581, 2013.

\bibitem{Barbera1995}
S.~Barbera and M.~O. Jackson.
\newblock {Strategy-Proof Exchange}.
\newblock {\em Econometrica}, 63(1):51--87, January 1995.

\bibitem{Blumrosen2014}
L.~Blumrosen and S.~Dobzinski.
\newblock Reallocation mechanisms.
\newblock In {\em ACM Conference on Electronic Commerce}, 2014.

\bibitem{bogomolnaia2004random}
A.~Bogomolnaia and H.~Moulin.
\newblock Random matching under dichotomous preferences.
\newblock {\em Econometrica}, 72(1):257--279, 2004.

\bibitem{Cai2012}
Y.~Cai, C.~Daskalakis, and S.~M. Weinberg.
\newblock Optimal multi-dimensional mechanism design: Reducing revenue to
  welfare maximization.
\newblock In {\em FOCS}, pages 130--139, 2012.

\bibitem{Cai2013}
Y.~Cai, C.~Daskalakis, and S.~M. Weinberg.
\newblock Understanding incentives: Mechanism design becomes algorithm design.
\newblock In {\em FOCS}, pages 618--627, 2013.

\bibitem{Dobzinski2012}
S.~Dobzinski and J.~Vondr{\'a}k.
\newblock The computational complexity of truthfulness in combinatorial
  auctions.
\newblock In {\em ACM Conference on Electronic Commerce}, pages 405--422, 2012.

\bibitem{Dughmi2010}
S.~Dughmi and T.~Roughgarden.
\newblock Black-box randomized reductions in algorithmic mechanism design.
\newblock In {\em FOCS}, pages 775--784, 2010.

\bibitem{Feldman2013}
M.~Feldman and Y.~Wilf.
\newblock Strategyproof facility location and the least squares objective.
\newblock In {\em ACM Conference on Electronic Commerce}, pages 873--890, 2013.

\bibitem{Fotakis13}
D.~Fotakis and C.~Tzamos.
\newblock Strategyproof facility location for concave cost functions.
\newblock In {\em ACM Conference on Electronic Commerce}, pages 435--452, 2013.

\bibitem{gal1972multiparametric}
T.~Gal and J.~Nedoma.
\newblock Multiparametric linear programming.
\newblock {\em Management Science}, 18(7):406--422, 1972.

\bibitem{Ghodsi2011}
A.~Ghodsi, M.~Zaharia, B.~Hindman, A.~Konwinski, S.~Shenker, and I.~Stoica.
\newblock Dominant resource fairness: Fair allocation of multiple resource
  types.
\newblock In {\em Proceedings of the 8th USENIX Conference on Networked Systems
  Design and Implementation}, NSDI'11, pages 24--24, 2011.

\bibitem{HartlineB2010}
J.~D. Hartline and B.~Lucier.
\newblock Bayesian algorithmic mechanism design.
\newblock In {\em STOC}, pages 301--310, 2010.

\bibitem{Lu2010}
P.~Lu, X.~Sun, Y.~Wang, and Z.~A. Zhu.
\newblock Asymptotically optimal strategy-proof mechanisms for two-facility
  games.
\newblock In {\em Proceedings of the 11th ACM Conference on Electronic
  Commerce}, EC '10, pages 315--324, 2010.

\bibitem{manea2007serial}
M.~Manea.
\newblock Serial dictatorship and pareto optimality.
\newblock {\em Games and Economic Behavior}, 61(2):316--330, 2007.

\bibitem{Miyagawa2001}
E.~Miyagawa.
\newblock Locating libraries on a street.
\newblock {\em Social Choice and Welfare}, 18(3):527--541, 2001.

\bibitem{Myerson81}
R.~B. Myerson.
\newblock Optimal auction design.
\newblock {\em Mathematics of Operations Research}, 6(1):58--73, 1981.

\bibitem{Myerson1983}
R.~B. Myerson and M.~A. Satterthwaite.
\newblock Efficient mechanisms for bilateral trading.
\newblock {\em Journal of Economic Theory}, 28:265--281, 10 1983.

\bibitem{Nisan99}
N.~Nisan and A.~Ronen.
\newblock Algorithmic mechanism design (extended abstract).
\newblock In {\em Proceedings of STOC}, pages 129--140, New York, NY, USA,
  1999. ACM.

\bibitem{nisan07a}
N.~Nisan, T.~Roughgarden, \'{E}va Tardos, and V.~V. Vazirani, editors.
\newblock {\em Algorithmic Game Theory}.
\newblock Cambridge University Press, 2007.

\bibitem{Papadimitriou2008}
C.~H. Papadimitriou, M.~Schapira, and Y.~Singer.
\newblock On the hardness of being truthful.
\newblock In {\em FOCS}, pages 250--259, 2008.

\bibitem{Parkes13}
D.~Parkes and S.~Seuken.
\newblock {\em Economics and Computation, book draft, to appear}.
\newblock Cambridge University Press.

\bibitem{Parkes2012}
D.~C. Parkes, A.~D. Procaccia, and N.~Shah.
\newblock Beyond dominant resource fairness: extensions, limitations, and
  indivisibilities.
\newblock In {\em ACM Conference on Electronic Commerce}, 2012.

\bibitem{Procaccia2013}
A.~D. Procaccia.
\newblock Cake cutting: not just child's play.
\newblock {\em Commun. ACM}, 56(7):78--87, 2013.

\bibitem{Procaccia2009}
A.~D. Procaccia and M.~Tennenholtz.
\newblock Approximate mechanism design without money.
\newblock In {\em ACM Conference on Electronic Commerce}, pages 177--186, 2009.

\bibitem{Ramchurn2012}
S.~D. Ramchurn, P.~Vytelingum, A.~Rogers, and N.~R. Jennings.
\newblock Putting the 'smarts' into the smart grid: a grand challenge for
  artificial intelligence.
\newblock {\em Commun. {ACM}}, 55(4):86--97, 2012.

\bibitem{Rogers2012}
A.~Rogers, S.~D. Ramchurn, and N.~R. Jennings.
\newblock Delivering the smart grid: Challenges for autonomous agents and
  multi-agent systems research.
\newblock In {\em Proceedings of the Twenty-Sixth {AAAI} Conference on
  Artificial Intelligence, July 22-26, 2012, Toronto, Ontario, Canada.}, 2012.

\bibitem{svensson1999strategy}
L.-G. Svensson.
\newblock Strategy-proof allocation of indivisible goods.
\newblock {\em Social Choice and Welfare}, 16(4):557--567, 1999.

\bibitem{Yang2014}
R.~Yang, B.~J. Ford, M.~Tambe, and A.~Lemieux.
\newblock Adaptive resource allocation for wildlife protection against illegal
  poachers.
\newblock In {\em International conference on Autonomous Agents and Multi-Agent
  Systems, {AAMAS} '14, Paris, France, May 5-9, 2014}, pages 453--460, 2014.

\end{thebibliography}

\end{document}